\documentclass[10pt]{article}
\usepackage{amsmath,amssymb,color,amsthm}
 \usepackage{amsmath}

\theoremstyle{definition}

\usepackage{amsmath,amssymb,color}
\usepackage{graphicx} 
\usepackage[mathscr]{euscript}
\usepackage{enumitem}
\usepackage{breqn}

\let\doilinkitem=\relax
\newcommand*{\doititle}[1]{%
  \@ifundefined{href}{#1}{%
    \@ifundefined{doilinkitem}{%
      #1%
    }{%
      \href{http://dx.doi.org/\doilinkitem}{#1}%
    }%
    \let\doilinkitem=\relax%
  }%
}

\usepackage{amsfonts}
\usepackage{multirow, verbatim}

\newcommand{\EQ}{\begin{equation}}
\newcommand{\EN}{\end{equation}}

\newcommand{\F}{ {\mathbb{F}}}

\newcommand{\hH}{\mathbf{H}}
\newcommand{\hh}{\mathbf{h}}

\newcommand{\WH}[1]{\mathcal{H}_{{#1}}}

\newcommand{\nN}{\mathbb{N}} 
\newcommand{\cC}{\mathbb{C}} 
\usepackage{todonotes}
\newtheorem{lem}{Lemma}

\newcommand{\vF}[2]{\mathbb{F}_{#1}^{#2}}
\newcommand{\fF}[2]{\mathbb{F}_{{#1}^{#2}}}

\newcommand{\Z}{{\rm Z\kern -0.31em Z}}
\newcommand{\ZZ}{{\rm Z\kern -0.31em Z}}

\allowdisplaybreaks
\newcommand{\mb}[1]{\mathbf{#1}}
\newcommand{\sB}[1]{\left( #1 \right)}
\def\whitebox{{\hbox{\hskip 1pt
        \vrule height 6pt depth 1.5pt
        \lower 1.5pt\vbox to 7.5pt{\hrule width
                  3.2pt\vfill\hrule width 3.2pt}%
        \vrule height 6pt depth 1.5pt
        \hskip 1pt } }}
\def\qed{\ifhmode\allowbreak\else\nobreak\fi\hfill\quad\nobreak\whitebox\medbreak}

\newtheorem{rem}{Remark}
\newtheorem{op}{Open Problem}
\newtheorem{defi}{Definition}

\newtheorem{cor}{Corollary}

\newtheorem{theo}{Theorem}
\newtheorem{ex}{Example}
\newtheorem{con}{Construction}

\usepackage{hyperref}

\title{Secondary constructions of vectorial $p$-ary weakly regular bent functions}

\author{Amar Bapi\'c}

\date{University of Primorska, IAM and FAMNIT}

\begin{document}

%
%
%
%
%
%
%
%
%

 \maketitle

\begin{abstract}
In \cite{Bapic, Tang, Zheng} a new method for the secondary construction of vectorial/Boolean bent functions via the so-called $(P_U)$ property was introduced. In 2018, Qi et al. generalized the methods in \cite{Tang} for the construction of $p$-ary weakly regular bent functions. The objective of this paper is to further generalize these constructions, following the ideas in \cite{Bapic, Zheng}, for secondary constructions of vectorial $p$-ary weakly regular bent and plateaued functions. We also present some infinite families of such functions via the $p$-ary Maiorana-McFarland class. Additionally, we give another characterization of the $(P_U)$ property for the $p$-ary case via second-order derivatives, as it was done for the Boolean case in \cite{Zheng}.
\end{abstract}

\maketitle


\section{Introduction}
An important class of Boolean functions, i.e. functions from $\vF{2}{n}\to\vF{2}{}$ (or $\fF{2}{n}\to\fF{2}{}$), was introduced in 1976 by Rothaus \cite{Rothaus}, which are defined in even number  of variables with maximum Hamming distance to the set of all affine functions. These functions are called bent functions and they have been studied exhaustively because of their applications to coding theory, cryptography, graph theory, etc. (for more details on bent functions see \cite{CarletBook, SihemBook}). In \cite{Kumar}, Kumar et al. generalized bent functions to functions over any arbitrary finite field. These functions have become of interest, because of ther connection with partial difference sets, strongly regular graphs, association schemes and orthogonal frequency-division multiplexing (OFDM) (see \cite{Che, Hyun, Pott, Tan}). For efficient wireless communication, generalized bent functions are used for large signal sets with low maximum crosscorrelation \cite{HellKum, Kim, Olsen, Zhou}. For some known constructions of $p$-ary bent functions, we refer to \cite{Buda, Helleseth, Mandal, Mes}.

In 2017, Tang et al. \cite{Tang} proposed a secondary construction of Boolean bent functions using the so-called $(P_U)$ property (cf. Definition \ref{def:PU}). Later, this method was extended to vectorial Boolean bent functions in \cite{Bapic, Zheng}. 

In \cite{TangPary}, the results of Tang et al. were generalized for the construction of several infinite families of $p$-ary weakly regular bent functions. In this paper, we further generalize these results to obtain a secondary construction of vectorial $p$-ary weakly regular bent and plateaued functions of the form $$F(x)=G(x)+\hH(x),$$
where $G$ is a suitable $p$-ary weakly regular bent $(n,m)$-function and $\hH$ is a $p$-ary $(n,t)$-function. Furthermore, we give a characterization of the $(P_U)$ property via second-order derivatives, as it was done for the Boolean case in \cite{Zheng}.

The rest of the paper is organized as follows. In Section \ref{sec:prel} we give some basic definitions and notations used throughout the paper. Our main construction of vectorial $p$-ary weakly regular bent and plateaued functions is presented in Section \ref{sec:con}. Some new infinite families of vectorial $p$-ary weakly regular bent functions via the $p$-ary Maiorana-McFarland class are presented in Section \ref{sec:ex}. We also show that certain monomial $p$-ary weakly regular bent $(n,m)$-functions cannot be used for this construction, as it was the case for $p=2$ in \cite{Bapic, Tang, Zheng}. Some concluding remarks are given in Section \ref{sec:concl}.

\section{Preliminaries}\label{sec:prel}
Let $\fF{p}{}=\{0,1,\ldots,p-1\}$ denote the prime field of characteristic $p>2$, and $\fF{p}{n}$ its extension field of degree $n$. The vector space $\vF{p}{n}$ is the space of all $n$-tuples $\mb{x}=(x_1,\ldots,x_n)$, where $x_i\in\fF{p}{}$ for $1\leq i\leq n$. With $\fF{p}{}^*$ we denote the multiplicative cyclic group consisting of $p^n-1$ elements. For convenience, we will sometimes identify the vector space $\fF{p}{n}$ with $\vF{p}{n}$. 

A function from $\fF{p}{n}$ to $\fF{p}{}$ is called a \textit{$p$-ary function} and the set of all such functions is denoted by $\mathcal{B}_n^p$. Any $p$-ary function $f\in\mathcal{B}_n^p$ can be uniquely expressed as a polynomial in $\fF{p}{}[x_1,\ldots,x_n]\setminus \langle x_1^p-x,\ldots,x_n^p-x\rangle$ as $$f(x_1,\ldots,x_n)=\sum_{\mb{a}\in\vF{p}{n}}\lambda_{\mb{a}}\prod_{i=1}^nx_i^{a_i},$$
where $\lambda_{\mb{a}}\in\fF{p}{}$. The algebraic degree of $f$ is defined as $$\deg(f)=\max\{wt(\mb{a}):\lambda_{\mb{a}}\neq 0\},$$
where $wt(\mb{a})=|\{i:a_i\neq 0,1\leq i\leq n\}|$ is the weight of $\mb{a}\in\vF{p}{n}$. The generalized Walsh-Hadamard transform (GWHT) and its inverse of a $p$-ary function $f\in\mathcal{B}_n^p$ at a point $\mb{a}\in\vF{p}{n}$ are defined by $$\WH{f}(\mb{a})=\sum_{\mb{x}\in\vF{p}{n}}\xi_p^{f(\mb{x})-\langle \mb{a},\mb{x}\rangle},$$
and
$$\xi^{f(\mb{a})}=p^{-n}\sum_{\mb{x}\in\vF{p}{n}}\WH{f}(x)\xi_p^{\langle \mb{a},\mb{x}\rangle},$$
respectively, where $\xi_p=e^{\frac{2\pi i}{p}}$ denotes the complex primitive $p$-th root of unity and $\langle \mb{a},\mb{b}\rangle $ denotes an inner product on $\vF{p}{n}$. For convenience, if we are considering functions in vector space notation, we will define $\langle \mb{x},\mb{y}\rangle=\sum_{i=1}^nx_iy_i$, and if are considering finite field notation, we will define $\langle \alpha,\beta \rangle=tr_n(\alpha\beta)$, where $$tr_m^n(\alpha):=\alpha+\alpha^{p^m}+\alpha^{p^{2m}}+\cdots +\alpha^{p^{m(n/m-1)}}$$
denotes the trace function from $\fF{p}{n}$ to $\fF{p}{m}$, $m|n$. For simplicity we will use the notation $tr_n:=tr_1^n$.

A function $f\in\mathcal{B}_n^p$ is said to be \textit{bent} if $|\WH{f}(\mb{a})|^2=p^n$ for all $\mb{a}\in\vF{p}{n}$. Furthermore, $f$ is said to be \textit{regular bent} if for every $b\in\vF{p}{n}$, $p^{-n/2}\WH{f}(b)=\xi_p^{f^*(b)}$ for some mapping $f^*\in\mathcal{B}_n^p$, which is then called the \textit{dual} of $f$. The bent function $f$ is said to be \textit{weakly regular} if there exists a complex number $z$ with $|z|=1$, such that $zp^{-n/2}\WH{f}(b)=\xi_p^{f^*(b)}$ for all $b\in\vF{p}{n}$. As noted in \cite{SihemBook}, regular bent functions can only be found for even $n$ and for odd $n$ with $p\mod 4=1$. Moreover, for a weakly regular bent function, the constant $z$ (defined above) can only be equal to $\pm 1$ or $\pm i$. Weakly regular bent functions always come in pairs, since the dual is bent as well. Moreover, it holds that $f^{**}(x)=f(-x),\ f^{***}(x)=f^{*}(-x),\ f^{****}(x)=f(x)$. For a $p$-ary function $f\in\mathcal{B}_n^p$, we define its derivative $D_af\in\mathcal{B}_n^p$ at a point $a\in\fF{p}{n}$ as $$D_af(x)=f(x+a)-f(x),\ x\in\fF{p}{n}.$$ Similarly, the $k$-th order derivative of $f$ with respect to $a_1,\ldots,a_k\in\fF{p}{n}$ is defined by $D_{a_1,\ldots,a_k}f(x)=D_{a_1}D_{a_2}\ldots D_{a_k}f(x)$, for all $x\in\fF{p}{n}$.

Any mapping $F$ from $\fF{p}{n}$ to $\fF{p}{m}$ is called a vectorial $p$-ary function (or a $p$-ary $(n,m)$-function). We say that $F$ is (weakly regular) bent if for every $u\in\fF{p}{m}^*$, its component function $F_u\in\mathcal{B}_n^p$ defined as $F_u(x)=tr_m(uF(x)),\ x\in\fF{p}{n},$ is  $p$-ary (weakly regular) bent. For $p=2$, these mappings are called Boolean bent $(n,m)$-functions and they exist for $m\leq n/2$ \cite{Nyberg}. On the other hand, if $p$ is an odd prime, $p$-ary bent $(n,m)$-functions exist for all $m\leq n$ \cite{AycaVPF}.

Another important class of vectorial $p$-ary bent functions are the so-called plateaued $p$-ary $(n,m)$-functions. Namely, we say that a $p$-ary $(n,m)$-function $F$ is plateaued if $\WH{F_{\lambda}}^2(x)\in\{0,p^{n+s}\}$ for all $x\in\fF{p}{n}$ and all $\lambda\in\fF{p}{m}^*$, for some $s\in\nN_0$ (which is called the amplitude). Specially, if $s=0$ we are talking about bent functions, and if $s=1$, then such functions are called near-bent. If all the components have the same amplitude $s$, then these functions are called $s$-plateaued $p$-ary $(n,m)$-functions.

Throughout the paper we will be using the notion of the so-called $(P_U)$ property, which is defined below.

\begin{defi}\label{def:PU}
Let $g\in\mathcal{B}_n^p$ be a $p$-ary bent function and let $U=\{u_1,\ldots,u_t\}\subset\fF{p}{n}$. We say that $g$ satisfies the $(P_U)$ property if $$g\sB{x+\sum_{i=1}^tw_iu_i}=g(x)+\sum_{i=1}^tw_ig_i(x),$$
for all $x\in\fF{p}{n}$ and $(w_1,\ldots,w_t)\in\vF{p}{t}$, for some $p$-ary functions $g_1,\ldots,g_t\in\mathcal{B}_n^p$.
\end{defi}

\section{Generic construction of vectorial $p$-ary bent functions}\label{sec:con}

Throughout the paper $G_{\lambda}^*$ denotes the dual of the $p$-ary bent component $G_{\lambda}$, $\lambda\in\fF{p}{m}^*$, of a vectorial $p$-ary bent function $G:\fF{p}{n}\to\fF{p}{m}$, $m|n$.

\begin{con}\label{con}
 Let $u_1,\ldots,u_t\in\fF{p}{n}^*$ be linearly independent elements over $\fF{p}{}$, where $m|n$ and $t| m$. Let $G:\fF{p}{n}\to\fF{p}{m}$ be any $p$-ary weakly regular bent function whose components $G_{\lambda}(x)=tr_m(\lambda {G}(x))$, with $ \lambda\in\fF{p}{m}^*$, satisfy
\begin{align}\label{property}
G^*_{\lambda}\sB{x+\sum_{i=1}^t u_iw_i}=G_{\lambda}^*(x) + \sum_{i=1}^t w_ig_i(x)
\end{align}
for all $x\in\fF{p}{n}$ and $(w_1,\ldots,w_t)\in\vF{p}{t}$, where $g_i(x)$ is a $p$-ary function from $\fF{p}{n}$ to $\fF{p}{}$, $1\leq i\leq t$. In this case we say that $G$ satisfies the \textit{property $(P_U)$}. 

Let $\hh(X_1,\ldots,X_t)$ be any vectorial $p$-ary function from $\vF{p}{t}$ to $\fF{p}{t}$. Define $F:\fF{p}{n}\to\fF{p}{m}$, using $G$ and $\hh$, as  
\begin{align}\label{newfunction}
F(x)=G(x)+\hH(x),
\end{align}
where $\hH:\fF{p}{n}\to\fF{p}{t}$ is defined by $\hH(x)=\hh(tr_n(u_1x),\ldots,tr_n(u_tx))$.  Equivalently, if $\hh$ is defined using the finite field notation so that $\hh: \fF{p}{t} \rightarrow \fF{p}{t}$, then define
\begin{align*}
F(x)=G(x)+\hH(x)=G(x)+\hh(tr_n(u_1x)+\alpha tr_n(u_2x)+\cdots + \alpha^{t-1}tr_n(u_tx)),
\end{align*}
where $\alpha$ is a primitive element of $\fF{p}{t}$.
\end{con}

\begin{rem}
We note that in this case $\hh$ can be a function in $\mathcal{B}_p^n$. This corresponds to the $p$-ary case of \cite[Theorem 3.3]{Zheng}.
\end{rem}

Using this construction, we prove the following result which is the $p$-ary equivalent of \cite[Theorem 1]{Bapic}.

\begin{theo}\label{thm:con2}
The function $F$ generated by Construction 1 is a $p$-ary weakly regular bent $(n,m)$-function.
\end{theo}

\begin{proof}
Let $\lambda\in\fF{p}{m}^*$ be arbitrary. Let us consider the component $G_{\lambda}$ and let $\hh_{\lambda}:\vF{p}{t}\to\fF{p}{t}$ be defined as $\hh_{\lambda}=tr_m(\lambda\hh)$. From the inverse GWHT, we have that 
\begin{align*}
\xi_p^{\hh_{\lambda}(X_1,\ldots,X_t)}=p^{-t}\sum_{(w_1,\ldots,w_t)\in\vF{p}{t}}\WH{\hh_{\lambda}}(w_1,\ldots,w_t)\xi_p^{\sum_{i=1}^t w_iX_i}.
\end{align*}
For any $x\in\fF{p}{n}$ and $1\leq i\leq t\leq m$, taking $X_i=tr_n(u_ix)$, we obtain 
\begin{align}\label{subst}
\xi_p^{\hh_{\lambda}(tr_n(u_1x),\ldots,tr_n(u_tx))}=p^{-t}\sum_{(w_1,\ldots,w_t)\in\vF{p}{t}}\WH{\hh_{\lambda}}(w_1,\ldots,w_t)\xi_p^{tr_n\sB{\sB{\sum_{i=1}^t w_iu_i}x}}.
\end{align}
Multiplying both sides of \eqref{subst} by $\xi^{G_{\lambda}(x)+tr_n(\beta x)}$, we have $$\xi^{G_{\lambda}(x)+\hH_{\lambda}(x)+tr_n(\beta x)}=p^{-t}\sum_{(w_1,\ldots,w_t)\in\vF{p}{t}}\WH{\hh_{\lambda}}(w_1,\ldots
,w_t)\xi_p^{G_{\lambda}(x)+tr_n\sB{\sB{\beta + \sum_{i=1}^t w_iu_i}x}}.$$
By summing the previous expression on both sides over all $x\in\fF{p}{n}$ and using the fact that $G$ is $p$-ary  weakly regular bent, we obtain that 
\begin{align}\label{WHTobtained}
\WH{F_\lambda }(\beta)&=p^{-t}\sum_{(w_1,\ldots,w_t)\in\vF{p}{t}}\WH{\hh_{\lambda}}(w_1,\ldots,w_t)\WH{G_\lambda}(\beta+\sum_{i=1}^t u_iw_i) \nonumber  \\
&=p^{-t+n/2}z\sum_{(w_1,\ldots,w_t)\in\vF{p}{t}}\WH{\hh_{\lambda}}(w_1,\ldots,w_t) \xi_p^{G^*_{\lambda}(\beta+\sum_{i=1}^t u_iw_i)},
\end{align}
where $z\in\cC$ such that $|z|=1$.
It follows from \eqref{property} and \eqref{WHTobtained} that $$\WH{F_{\lambda}}(\beta)=p^{-t+n/2}z\xi^{G_{\lambda}^*(\beta)}\sum_{(w_1,\ldots,w_t)\in\vF{p}{t}}\WH{\hh_{\lambda}}(w_1,\ldots,w_t) \xi_p^{\sum_{i=1}^t w_ig_{i}(\beta)}.$$

The sum on the right-hand side corresponds to the inverse GWHT {of $\hh_{\lambda}$ at the point $(g_1(\beta),\ldots,g_t(\beta))$} and thus we have $$\WH{F_{\lambda}}(\beta)=p^{n/2}z\xi_p^{G_{\lambda}^*(\beta)+\hh_{\lambda}(g_1(\beta),g_2(\beta),\ldots,g_t(\beta))}.$$
Since $\beta\in\fF{p}{n}$ is arbitrary, we have that $F_{\lambda}$ is  $p$-ary weakly regular bent for all $\lambda\in\fF{p}{m}^*$. In other words, $F$ is a  $p$-ary weakly regular bent $(n,m)$-function.
\end{proof}

\begin{rem}\label{rem:classcard}
If we have a function $f:X\to Y$, then the number of possible functions $f$ equals to $\#Y^{\#X}$. Thus, since $\hh$ is a $p$-ary $(t,t)$-function, there are $p^{tp^t}$ possible choices for $\hh$. Hence, we can construct at most $p^{tp^t}$ $p$-ary bent $(n,m)$-functions $F$ from a fixed bent function $G$ and an arbitrary function $\hh$. {In the case when} $p=3,n=4$ and $m=t=2$, we have $3^{18}$ possibilities. 
\end{rem}

Similarly, as noted in \cite{Zheng} for the Boolean case, we can use Construction 1 to obtain new instances of plateaued $p$-ary $(n,m)$-functions.

\begin{cor}\label{cor:plate}
With the same conditions as in Theorem \ref{thm:con2}, let $l$ be any positive integer. Let $\hh_i$ be any reduced polynomial in $\fF{p}{}[X_1,\ldots,X_n]$, for $i=1,\ldots,l$. Then $$F(x)=(G(x),\hh_1(tr_n(u_1x),\ldots,tr_n(u_tx)),\ldots,\hh_l(tr_n(u_1x),\ldots,tr_n(u_tx)))$$
is a plateaued $p$-ary $(n,m+l)$-function if and only if the $p$-ary $(n,l)$-function $x\mapsto (\hh_1(tr_n(u_1x),\ldots,tr_n(u_tx)),\ldots,\hh_l(tr_n(u_1x),\ldots,tr_n(u_tx)))$, $x\in\fF{p}{n}$, is plateaued.
\end{cor}

\begin{proof}
For any $v\in\vF{p}{t}$ the function $\langle v,(\hh_1,\ldots,\hh_l)\rangle$ is again a reduced polynomial, and thus by Theorem 1, the $p$-ary function $\langle (\lambda,v),F\rangle$ is bent for all $\lambda\in\fF{p}{m}^*$. Hence, $F$ is plateaued if and only if all the components $\langle (0,v),(\hh_1,\ldots,\hh_l)\rangle$ are plateaued for $v\neq 0$, or equivalently, if $x\mapsto (\hh_1,\ldots,\hh_l)$, $x\in\fF{p}{n}$, is a plateaued $p$-ary $(n,l)$-function.
\end{proof}

Before providing instances of new vectorial $p$-ary weakly regular bent functions, we will provide another characterisation of the $(P_U)$ property via second-order derivatives, as it was done by Zheng et. al for the binary case in \cite{Zheng}.

\begin{lem}\label{lem:der}
Let $g\in\mathcal{B}_p^n$ be any $p$-ary weakly regular bent function. Then the following statements are equivalent.
\begin{enumerate}[label=(\roman*)]
    \item There exist $u_1,\ldots,u_t\in\fF{p}{n}$ and $g_1,\ldots,g_t\in\mathcal{B}_p^n$ such that
    \begin{align}\label{property0}
        g\sB{x+\sum_{i=1}^t w_iu_i}=g(x)+\sum_{i=1}^t w_ig_i(x),
    \end{align}
    for all $\mb{w}=(w_1,\ldots,w_t)\in\fF{p}{t}$.
    \item $D_{u_i}D_{u_j}g\equiv 0$ for all $1\leq i,j\leq t$.
\end{enumerate}
\end{lem}

\begin{proof}
$(i\Rightarrow ii)$ As $\mb{w}$ is arbitrary, let us take $\mb{w}=\mb{e}_i$, where $\mb{e}_i=(e_0,\ldots,e_t)$ with $e_k=1$ for $k=i$, and $e_k=0$, otherwise, for $1\leq k\leq t$. Then \eqref{property0} becomes $g(x+u_i)=g(x)+g_i(x)$, or equivalently, $$g_i(x)=g(x+u_i)-g(x)=D_{u_i}g(x).$$
Similarly, for any $1\leq i,j\leq t$, we deduce that $$g(x+u_i+u_j)=g(x)+g_i(x)+g_j(x)=-g(x)+(g(x+u_i)+g(x+u_j))\Rightarrow D_{u_i}D_{u_j}g(x)=0,$$
for any $x\in\fF{p}{n}$.

$(ii\Rightarrow i)$ Let us define $g_i:=D_{u_i}g$, for $i=1,\ldots,t$. Let $q\in\fF{p}{}$ and $1\leq i\leq t$ be arbitrary. We will show that that $g(x+qu_i)=g(x)+qD_{u_i}g(x)$, for all $x\in\fF{p}{n}$. From the assumption that $D_{u_i}D_{u_j}g\equiv 0$ and taking $i=j$, we have that 
\begin{align}\label{eq:deri2}
    &\ g(x+2u_i)-2g(x+u_i)+g(x)=0 \nonumber \\
    \Rightarrow &\  g(x+2u_i)=-g(x)+2g(x+u_i)=g(x)+2D_{u_i}g(x).
\end{align}
If we change $x$ with $x+u_i$ in \eqref{eq:deri2}, then:
\begin{align}\label{eq:deri3}
    g(x+3u_i)-2g(x+2u_i)+g(x+u_i)=0.
\end{align}
Furthermore, from  \eqref{eq:deri2} and \eqref{eq:deri3}, we also note that
\begin{align}\label{eq:deri3l}
    g(x+3u_i)&=2(g(x)+2D_{u_i}g(x))-g(x+u_i)\nonumber\\
    &=g(x)+4D_{u_i}g(x)-D_{u_i}g(x)\nonumber\\
    &=g(x)+3D_{u_i}g(x)
\end{align}
If we continue inductively, we observe that $g(x+qu_i)=g(x)+3D_{u_1}g(x)$ holds indeed for all $x\in\fF{p}{n}$ and all $q\in \fF{p}{}$. Assume now that $1\leq i,j\leq t$ and $w_i,w_j\in\fF{p}{}$ are arbitrary. Since $g(x+qu_i)=g(x)+qD_{u_i}g(x)$ holds for all $q\in\fF{p}{}$ and $1\leq i\leq t$, we have that
\begin{align*}
    g(x+w_iu_i+w_ju_j)&=g((x+w_iu_i)+w_ju_j)=g(x+w_iu_i)+w_jD_{u_j}g(x+w_iu_i)\\
    &=g(x)+w_iD_{u_i}g(x)+w_j\sB{g(x+w_iu_i+u_j)-g(x+w_iu_i)}\\
    &=g(x)+w_iD_{u_i}g(x)+w_j\left(g(x+u_j)+w_iD_{u_i}g(x+u_j)\right.\\
    &\left.-g(x)-w_iD_{u_i}g(x)\right)\\
    &=g(x)+w_iD_{u_i}g(x)+w_j\sB{D_{u_j}g(x)+w_iD_{u_j}D_{u_i}g(x)}\\
    &=g(x)+w_iD_{u_i}g(x)+w_jD_{u_j}g(x)
\end{align*}
Using mathematical induction, it is easy to show that \eqref{property0} holds for all $(w_1,\ldots,w_t)\in\fF{p}{t}$ and all $x\in\fF{p}{n}$.
\end{proof}

\begin{rem}
We note that the functions $g_i$ in \eqref{property0} are exactly the derivatives $D_{u_i}g$, $i=1,\ldots,t$.
\end{rem}

\section{New infinite families of vectorial $p$-ary weakly regular bent functions}\label{sec:ex}

Using similar methods as in \cite{Bapic, TangPary, Tang, Zheng}, we will present certain classes of $p$-ary vectorial weakly regular bent functions whose components satisfy the property $(P_U)$ and thus may be used for the construction of new $p$-ary weakly regular bent and plateaued functions via Construction \ref{con}. We note that the proofs are analogous to the binary case proved in \cite{Bapic}. Before that, we present some observations on certain monomial weakly regular bent functions and their connection to the $(P_U)$ property.

\subsection{Observations on monomial $p$-ary weakly regular bent functions}

Let $n=2m$. Since the functions of the form $x\mapsto tr_m(\lambda x^{p^m+1}),\ x\in\fF{p}{n}$ are $p$-ary weakly regular bent for $\lambda\in\fF{p}{m}^*$ (see e.g. \cite{Liu}), the function $F$ defined by $G(x)=x^{p^m+1}$ is a vectorial $p$-ary weakly regular bent function. From \cite{Helleseth}, we know that the dual of $G_{\lambda}$ is defined as $G_{\lambda}^*(x)=-tr_m\left(\frac{x^{p^{m}+1}}{\lambda^{p^m}+\lambda}\right)$. From \cite[Theorem 3.4]{TangPary}, the component $G_{\lambda}^*$ satisfies the property $(P_U)$ if $$tr_m\left(\frac{u_i^{p^{m}}u_j+u_iu_j^{p^m}}{\lambda^{p^m}+\lambda}\right)=0,$$
for all $u_i,u_j\in U\subseteq \fF{p}{n}$ with $|U|=t|m$. Thus, for $i=j$, we must have that $tr_m(2(\lambda^{p^m}+\lambda)^{-1}u_i^{p^m+1})=0$ for all $1\leq i\leq t$. If we want to construct a vectorial $p$-ary bent function $F$ from $G$ via Construction \ref{con}, then the above equality has to hold for all $\lambda\in\fF{p}{m}^*$. Thus, we must have that $u_i^{p^m+1}=0$ and consequently $u_i=0$ for all $1\leq i\leq t$. Hence, one cannot construct a vectorial $p$-ary weakly regular bent $(n,m)$-function via the Kasami function $G$. 

Similarly, let us we consider the function $G(x)=x^2$ on $\fF{p}{n}$. The duals of its components $G_{\lambda}(x)=tr_n(\lambda x^2)$ are defined by $G_{\lambda}^*(x)=-tr(\frac{x^2}{4\lambda})$, for $\lambda\in\fF{p}{n}^*$. Suppose $U\subset \fF{p}{n}$ denotes the set from Construction \ref{con}. From Lemma \ref{lem:der}, we must have $D_uD_vG_{\lambda}^*\equiv 0$, for all $u,v\in U$, and consequently, for all $\lambda\in\fF{p}{n}^*$. Let $u,v\in U$ be arbitrary, then
\begin{align*}
D_uD_vG_{\lambda}^*(x)&=-tr_n\sB{\frac{x^2-(x+u)^2-(x+u)^2+(x+2u)^2}{4\lambda}}\\
&=-tr\sB{\frac{2uv}{4\lambda}}=-tr\sB{\frac{uv}{2\lambda}}.
\end{align*}
Thus, we must have that $tr\sB{\frac{uv}{2\lambda}}=0$ for all $u,v\in U$ and all $\lambda\in\fF{p}{n}^*$. Specially, if $u=v$, then we have that $tr\sB{\frac{u^2}{2\lambda}}=0$. However, this is only possible if $u=0$. In other words, we cannot construct vectorial $p$-ary bent functions via $G$ and Construction \ref{con}. Consequently, we have the following remark.

\begin{rem}
If $G$ is a vectorial $p$-ary weakly regular bent function defined as $G(x)=x^{p^m}+1$ or $G(x)=x^2$, for $x\in\fF{p}{n}$, then one cannot construct new vectorial $p$-ary weakly regular bent functions via Construction \ref{con}.
\end{rem}

Based on this observation, we have the following interesting open problem.

\begin{op}
Can we find an exponent $d$ such that $G(x)=x^d$ is a monomial $p$-ary weakly regular bent function and all of its components satisfy the $(P_U)$ property?
\end{op}

\subsection{New infinite families of vectorial $p$-ary weakly regular bent functions from the $p$-ary Maiorana-McFarland class}

Let $n=2m$ and let us identify $\fF{p}{n}$ with $\fF{p}{m}\times\fF{p}{m}$. The well known Maiorana-McFarland class of vectorial $p$-ary bent functions can be defined as $$F(x,y)=x\pi(y)+g(y),\ x,y\in\fF{p}{m},$$
where $\pi:\fF{p}{m}\to\fF{p}{m}$ is a permutation and $g\in\mathcal{B}_p^n$ is an arbitrary $p$-ary function. Let $\lambda\in\fF{p}{m}^*$ be arbitrary, we then have the component $F_{\lambda}(x,y)=tr_m(\lambda x\pi(y)+\lambda g(y))$. Its corresponding dual is defined with (see \cite{Ayca}): $$F_{\lambda}^*(x,y)=tr_m\sB{-y\pi^{-1}(x/\lambda)+\lambda g(\pi^{-1}(x/\lambda)},$$ where $\pi^{-1}$ is the inverse permutation of $\pi$. Following the methodology in \cite{TangPary}, we note that for $\alpha=(a_1,a_2),\beta=(b_1,b_2)\in\fF{p}{m}\times\fF{p}{m}$, the scalar product $tr_m(\alpha\beta)$ can be defined as $tr_m(a_1b_1+a_2b_2)$.

In \cite{TangPary} the authors considered the $p$-ary case for linearized polynomials. In the following results, we extend this notion to the vectorial $p$-ary case and obtain new instances of vectorial $p$-ary weakly regular bent and plateaued functions. The vectorial $p$-ary function \eqref{newfunction} in Construction \ref{con} can be rewritten in bivariate form as:
\begin{align*}
F(x,y)=G(x,y)+\hh\sB{tr_m(\alpha_ix+\beta_iy),\ldots,tr_m(\alpha_tx+\beta_ty)},
\end{align*}
where the elements $u_i\in U$ correspond to $u_i=(\alpha_i,\beta_i)\in\fF{p}{m}\times\fF{p}{m}$.

\begin{lem}\label{lem:MM1}
Let $n=2m$ and $u_1,\ldots,u_t\in\fF{p}{n}^*$ be linearly independent elements over $\fF{p}{}$, where $1\leq t|m$. Denote $u_i=(\alpha_i,\beta_i)\in\fF{p}{m}\times\fF{p}{m}$. Let $G(x,y)=y\pi(x)$, where $\pi$ is a linear permutation over $\fF{p}{m}$. If $tr_m\sB{\beta_i\pi^{-1}\sB{\frac{\alpha_j}{\lambda}}+\beta_j\pi^{-1}\sB{\frac{\alpha_i}{\lambda}}}=0$ for each $1\leq i,j \leq t$ and $\lambda\in\fF{p}{m}^*$, then the dual component $G_{\lambda}^*$ satisfies \eqref{property} with 
\begin{align}\label{eq:MM}
g_i(x,y)= -tr_m\sB{y\pi^{-1}\sB{\frac{\alpha_i}{\lambda}}+\beta_i\pi^{-1}\sB{\frac{x}{\lambda}}}.
\end{align}
\end{lem}

\begin{proof}
Let $X=x+\sum_{i=1}^tw_i\alpha_i$ and $Y=y+\sum_{i=1}^tw_i\beta_i$. It follows from \eqref{property} and the fact that $\pi$ is linear that 
\begin{align*}
G_{\lambda}^*\sB{X,Y}&=tr_m\sB{-\sB{y+\sum_{i=1}^tw_i\beta_i}\pi^{-1}\sB{\frac{x}{\lambda}+\sum_{i=1}^tw_i\frac{\alpha_i}{\lambda}}}\\
&=G_{\lambda}^*(x,y)-\sum_{i=1}^tw_i tr_m\sB{y\pi^{-1}\sB{\frac{\alpha_i}{\lambda}}+\beta_i\pi^{-1}\sB{\frac{x}{\lambda}}}\\
&-\sum_{i=1}^t w_i^2tr_k\sB{\beta_i\pi^{-1}\sB{\frac{\alpha_i}{\lambda}}}\\
&-\sum_{1\leq i<j\leq t} w_iw_jtr_m\sB{\beta_i\pi^{-1}\sB{\frac{\alpha_j}{\lambda}}+\beta_j\pi^{-1}\sB{\frac{\alpha_i}{\lambda}}}\\
&=G_{\lambda}^*(x,y)+\sum_{i=1}^tw_i g_i(x,y)-\sum_{i=1}^t w_i^2tr_m\sB{\beta_i\pi^{-1}\sB{\frac{\alpha_i}{\lambda}}}-\\
&\sum_{1\leq i<j\leq t} w_iw_jtr_m\sB{\beta_i\pi^{-1}\sB{\frac{\alpha_j}{\lambda}}+\beta_j\pi^{-1}\sB{\frac{\alpha_i}{\lambda}}},
\end{align*}
where $g_i$ is defined by \eqref{eq:MM}. The conclusion follows from the assumption that $$tr_m\sB{\beta_i\pi^{-1}\sB{\frac{\alpha_j}{\lambda}}   +\beta_j\pi^{-1}\sB{\frac{\alpha_i}{\lambda}}}=0,$$for each $1\leq i,j \leq t$ and $\lambda\in\fF{p}{m}^*$.
 \end{proof}

The following result is an immediate consequence of Lemma \ref{lem:MM1}.

\begin{cor}\label{cor:MM1} Let $\alpha_1,\ldots,\alpha_t\in\fF{p}{m}^*$ be linearly independent elements over $\fF{p}{}$, $1\leq t\leq k$. Denote  $u_i=(\alpha_i,0)$ and let $G(x,y)=y\pi(x)$, $x,y\in\fF{p}{m}$, where $\pi$ is a linear permutation over $\fF{p}{m}$. Then, the dual component $G_{\lambda}^*$ satisfies \eqref{property} with 
\begin{align*}
g_i(x,y)= tr_m\sB{y\pi^{-1}\sB{\frac{\alpha_i}{\lambda}}},
\end{align*}
for any $\lambda\in\fF{p}{m}^*$.
\end{cor}

Thus, as an immediate result of Theorem \ref{thm:con2} (Corrolary \ref{cor:plate}) and Corollary \ref{cor:MM1}, we have the following infinite family of vectorial $p$-ary weakly regular bent (plateaued) $(2m,m)$-functions.

\begin{theo}\label{thm:MM1} Let $\alpha_1,\ldots,\alpha_t\in\fF{p}{m}^*$ be linearly independent elements over $\fF{p}{}$, $t| m$. Let $G(x,y)=y\pi(x)$, where $\pi$ is a linear permutation over $\fF{p}{m}$, and let $\hh$ be any vectorial function from $\vF{p}{t}$ to $\fF{p}{t}$. Then, the function $F:\F_{p^m} \times \F_{p^m} \rightarrow \F_{p^m}$ given by  $$F(x,y)=y\pi(x)+\hh(tr_m(\alpha_1x),\ldots,tr_m(\alpha_tx)),$$ generated by Construction \ref{con}, is a  vectorial $p$-ary weakly regular bent $(n,m)$-function.
\end{theo}

\begin{theo}\label{thm:MM1_2} Let $\alpha_1,\ldots,\alpha_t\in\fF{p}{m}^*$ be linearly independent elements over $\fF{p}{}$, $t| m$. Let $G(x,y)=y\pi(x)$, where $\pi$ is a linear permutation over $\fF{p}{m}$, and let $\hh_i$ be any reduced polynomial in $\fF{p}{}[X_1,\ldots,X_t]$, for $1\leq i\leq l$, such that the function $$x\mapsto \hH(x)=(\hh_1(tr_m(\alpha_1x),\ldots,tr_m(\alpha_tx)),\ldots,\hh_l(tr_m(\alpha_1x),\ldots,tr_m(\alpha_tx))),$$ with $x\in\fF{p}{m}$, is a plateaued $p$-ary $(m,l)$-function. Then, the function $F:\F_{p^m} \times \F_{p^m} \rightarrow \vF{p}{m+l}$ defined by  $$F(x,y)=(y\pi(x),\hH(x))$$ is a plateaued $p$-ary $(n,m+l)$-function.
\end{theo}

\begin{ex}
Let $G:\fF{3}{4}\times\fF{3}{4}\to\fF{3}{4}$ be defined with $G(x,y)=xy$. Let $U=\{1,\beta,\beta^3,\beta^9\}$, where $\beta=\alpha^{82}$, and  $\alpha$ be a root of the primitive polynomial $p(x)=x^8 + 2x^5 + x^4 + 2x^2 + 2x + 2\in \fF{3}{8}[x]$. Let $\hh:\fF{3}{4}\to\fF{3}{4}$ be defined with $\hh(X)=X^{13}$. From Theorem \ref{thm:MM1}, the function $$F(x,y)=xy+\sB{tr_4(x)+\beta tr_4(\beta x)+\beta^2tr_4(\beta^2x)+\beta^3tr_4(\beta^3x)}^{13}$$ is a ternary weakly regular bent $(8,4)$-function.
\end{ex}

\section{Concluding remarks}\label{sec:concl}
In this paper, we generalized the work of Qi et al.
to vectorial $p$-ary weakly regular bent and plateaued functions. We obtained a new infinite family of vectorial $p$-
ary weakly regular bent and plateaued functions from the $p$-ary Maiorana-McFarland class. We also generalized the $(P_U)$ property via second-order derivatives and gave some observations on the connection between monomial $p$-ary weakly regular bent functions and the $(P_U)$ property.

\end{document}